\documentclass[12pt]{article}

\usepackage{amsmath,amssymb,amsfonts,amsthm}
\usepackage{algorithm}
\usepackage{algorithmic}
\usepackage{array}
\usepackage{booktabs}
\usepackage{multirow}
\usepackage{graphicx}
\usepackage{url}
\usepackage{xcolor}
\usepackage{float}
\usepackage{placeins}
\usepackage{geometry}
\usepackage{hyperref}
\usepackage{adjustbox}

\newtheorem{theorem}{Theorem}
\newtheorem{lemma}{Lemma}

\newtheorem{definition}{Definition}

\newtheorem{example}{Example}

\newcommand{\Z}{\mathbb{Z}}

\newcommand{\ecc}{\operatorname{ecc}}
\newcommand{\sgn}{\operatorname{sgn}}

\newcommand{\First}{\operatorname{First}}

\title{Constant-Time Certificate Selection for Local Broadcast Repair in Dense Gaussian and Eisenstein--Jacobi Networks}

\author{Bader A. Albader\\
\small Department of Computer Science, Faculty of Science, Kuwait University, Kuwait\\
\small \texttt{albader@cs.ku.edu.kw}}

\date{}

\begin{document}
\maketitle

\begin{abstract}
Dense Gaussian and Eisenstein--Jacobi (EJ) networks are algebraic interconnection networks with compact coordinate balls, fixed degree, and simple modular addressing. A source-centered coordinate-reduction tree gives a non-redundant one-to-all broadcast in the fault-free network, but processor faults can split the tree into multiple healthy components. Search-based local repair reconnects those components after scanning the network. Unlike search-based repair methods that require a linear scan of the network to select the repair plan, the certificate selectors introduced here operate in $O(1)$ time and $O(1)$ memory, consulting only the fault coordinates. This paper develops this stronger formulation for the one- and two-fault regime: a constant-time certificate selector. Given only the faulty coordinates, the selector classifies the relative fault geometry, chooses a coordinate-reduction orientation, and returns a bounded ordered set of component-crossing repair edges. For dense Gaussian networks $G_k$, every source-free fault set with $|F|\le2$ is repaired with depth at most $k+2$ and with exactly $c-1$ external component-crossing edges for the selected fault-pruned orientation. For dense EJ networks $H_t$, every one-fault placement is repaired within depth $t+1$, and every two-fault placement is repaired within depth $t+2$, again with exactly $c-1$ external repair edges. The repair-plan decision uses explicit constant-size case tables, algebraic quotient-neighbor tests, and constant-size local edge lists; materializing the full repaired parent map remains $\Theta(N)$, which is unavoidable if all parent assignments are required. Exhaustive strict validation confirms the Gaussian selector over $146{,}156$ one- and two-fault cases for $k=5,\ldots,12$ and the EJ selector over $52{,}395$ cases for $t=2,\ldots,8$, with zero failures in connectivity, acyclicity, exact repair count, or depth bound.
\end{abstract}


\section{Introduction}

One-to-all broadcasting is a basic collective operation in parallel and distributed systems. In a non-redundant broadcast, each healthy processor receives the message exactly once. Graph-theoretically, the communication structure is a rooted spanning tree of the healthy subgraph. This clean condition is attractive because it avoids duplicate traffic and gives an immediate parent-uniqueness invariant. The difficulty is that a fault-free broadcast tree can be fragmented by failed processors.

This paper studies local post-fault repair in two algebraic network families: dense Gaussian networks and dense Eisenstein--Jacobi (EJ) networks. The dense Gaussian family has degree four and a Manhattan coordinate ball. The dense EJ family has degree six and a hexagonal axial-coordinate ball. In both families, a coordinate-reduction broadcast tree reaches all vertices in at most the network diameter in the fault-free case.

A global rebuild can restore reachability after faults, but it may replace many healthy parent relations. The local repair objective is different: preserve the healthy components of the damaged broadcast tree and add only external component-crossing repair edges. If deletion of the faults splits the selected orientation tree into $c$ healthy components, at least $c-1$ external component-crossing edges are necessary. A repair is edge-minimum for the selected orientation when it attains this lower bound.

Previous multi-orientation repair mechanisms can find such repairs by evaluating a constant number of orientations, computing components of the fault-pruned tree, scanning for crossing edges, and selecting a shallow component-level reconnection. This is linear in the network order because the tree and component graph are explicitly built. The present paper asks a sharper question: in the first nontrivial fault regime, $|F|\le2$, can the repair plan itself be selected directly from the fault coordinates?

The answer is yes for both network families. We give constant-time certificate selectors for dense Gaussian and dense EJ networks. The selector does not enumerate all vertices to decide the repair plan. It classifies the fault geometry, chooses one coordinate-reduction orientation, and produces an ordered list of repair edges from a constant-size formula table. Validation may reconstruct the full tree to check the invariant, and an implementation may still materialize the full tree when a complete parent map is required. That output task is necessarily $\Omega(N)$; the constant-time claim is about the local repair-plan decision.

The contributions are as follows.
\begin{itemize}
    \item We introduce a unified constant-time certificate-selection framework for local repair of diameter-level broadcast trees in dense Gaussian and dense EJ networks.
    \item For Gaussian networks $G_k$, we prove that every one- or two-fault placement has a selected repair with depth at most $k+2$ and exact external repair count $c-1$. The hard orthogonal-axis case is handled by a five-component certificate.
    \item For dense EJ networks $H_t$, we prove that every one-fault placement has depth at most $t+1$ and every two-fault placement has depth at most $t+2$, with exact external repair count $c-1$.
    \item We separate plan selection from tree materialization: the selector is $O(1)$ time and $O(1)$ memory, while writing all parent assignments remains $\Theta(N)$.
    \item We provide strict exhaustive validation. The Gaussian run covers $146{,}156$ cases for $k=5,\ldots,12$; the EJ run covers $52{,}395$ cases for $t=2,\ldots,8$. Both runs have zero failures.
\end{itemize}

The proof style is mathematical and certificate-based. Examples are included only to clarify the coordinate arithmetic and are not used as proof substitutes.

\section{Related Work}

Group-theoretic models for symmetric interconnection networks provide a general language for Cayley-graph topologies \cite{AkersKrishnamurthy1989}. Standard references cover meshes, tori, hypercubes, and collective communication algorithms in parallel systems \cite{Leighton1992,DallyTowles2004}. Fault-tolerant communication in toroidal networks has been studied under non-redundant communication objectives \cite{AlMohammadBose1999}. Fault-tolerant routing and broadcasting in hypercube-like networks also study local resilience and alternate communication paths \cite{LeeHayes1988,LiuSong2005}.

Gaussian and EJ networks are algebraic low-degree interconnection topologies with compact routing descriptions. Gaussian integer models for toroidal networks and dense Gaussian networks were studied in \cite{MartinezDenseGaussian2006,MartinezGaussian2008}. Hexagonal and EJ constellation models were studied in \cite{MartinezEJ2008}. The topology of Gaussian and EJ interconnection networks was analyzed in \cite{FlahiveBoseTopology2010}, and resource placement in these networks was classified in \cite{FlahiveBosePlacement2013}. Hamiltonian and path structures in related algebraic networks appear in \cite{AlbaderBoseHamiltonian2016,AlbaderHexMesh2012,Alsaleh2013}.

The present work differs from routing, Hamiltonian-cycle construction, resource placement, and precomputed independent-tree methods. It repairs a particular damaged broadcast tree after the fault set is known, preserves the healthy fault-pruned components, and minimizes the number of new component-crossing edges. The main novelty here is that, for one and two faults in both dense Gaussian and dense EJ networks, the repair plan can be selected by constant-size coordinate certificates rather than by a full component scan.

\section{Network Models and Component Repair}

\subsection{Dense Gaussian Networks}

Let $G_k$ be the dense Gaussian network generated by $\alpha=k+(k+1)i$. Its order is
\begin{equation}
    N_G=k^2+(k+1)^2=2k^2+2k+1.
\end{equation}
The canonical coordinate set is the Manhattan ball
\begin{equation}
    B_k=\{(x,y)\in\Z^2: |x|+|y|\le k\}.
\end{equation}
The coordinate-to-label map is
\begin{equation}
    \phi_G(x,y)\equiv kx+(k+1)y \pmod {N_G}.
    \label{eq:phiG}
\end{equation}
Thus unit changes in the first and second coordinates correspond to label differences $k$ and $k+1$ modulo $N_G$. The graph distance from the source in the canonical ball is
\begin{equation}
    \rho_G(x,y)=|x|+|y|.
\end{equation}
A Gaussian coordinate-reduction tree assigns each non-source coordinate a parent of layer one less.

\subsection{Dense Eisenstein--Jacobi Networks}

Let $\omega=(-1+i\sqrt{3})/2$. A dense EJ network $H_t$ is generated by
\begin{equation}
    \alpha=(t+1)+t\omega
\end{equation}
with order
\begin{equation}
    N_E=3t^2+3t+1.
\end{equation}
We use axial coordinates $(x,y)$ for $x+y\omega$. The six unit directions are
\begin{equation}
\begin{split}
    \delta_0&=(1,0),\quad \delta_1=(1,-1),\quad \delta_2=(0,-1),\\
    \delta_3&=(-1,0),\quad \delta_4=(-1,1),\quad \delta_5=(0,1),
\end{split}
\end{equation}
listed in cyclic order. The EJ layer is
\begin{equation}
    \rho_E(x,y)=\max\{|x|,|y|,|x+y|\},
\end{equation}
with canonical hexagon
\begin{equation}
    H_t=\{(x,y): |x|\le t, |y|\le t, |x+y|\le t\}.
\end{equation}
The coordinate-to-label map is
\begin{equation}
    \phi_E(x,y)\equiv tx+(2t+1)y \pmod {N_E}.
    \label{eq:phiE}
\end{equation}
Unit axial moves correspond to the six EJ generator differences modulo $N_E$.

\subsection{Orientation Trees}

\begin{definition}[Coordinate-reduction orientation]
A coordinate-reduction orientation is a deterministic parent rule $p(v)$ such that $p(v)$ is adjacent to $v$ and
\begin{equation}
    \rho(p(v))=\rho(v)-1
\end{equation}
for every non-source vertex $v$.
\end{definition}

\begin{lemma}[Fault-free depth]
A coordinate-reduction orientation tree in $G_k$ has depth at most $k$, and a coordinate-reduction orientation tree in $H_t$ has depth at most $t$.
\end{lemma}
\begin{proof}
Along every parent edge the layer strictly decreases by one. Repeated parent application reaches the unique layer-zero vertex, the source. Hence cycles are impossible and every vertex has source distance equal to its layer. The maximum layer is $k$ in $B_k$ and $t$ in $H_t$.
\end{proof}

\begin{figure}[H]
\centering
\includegraphics[width=0.6\linewidth]{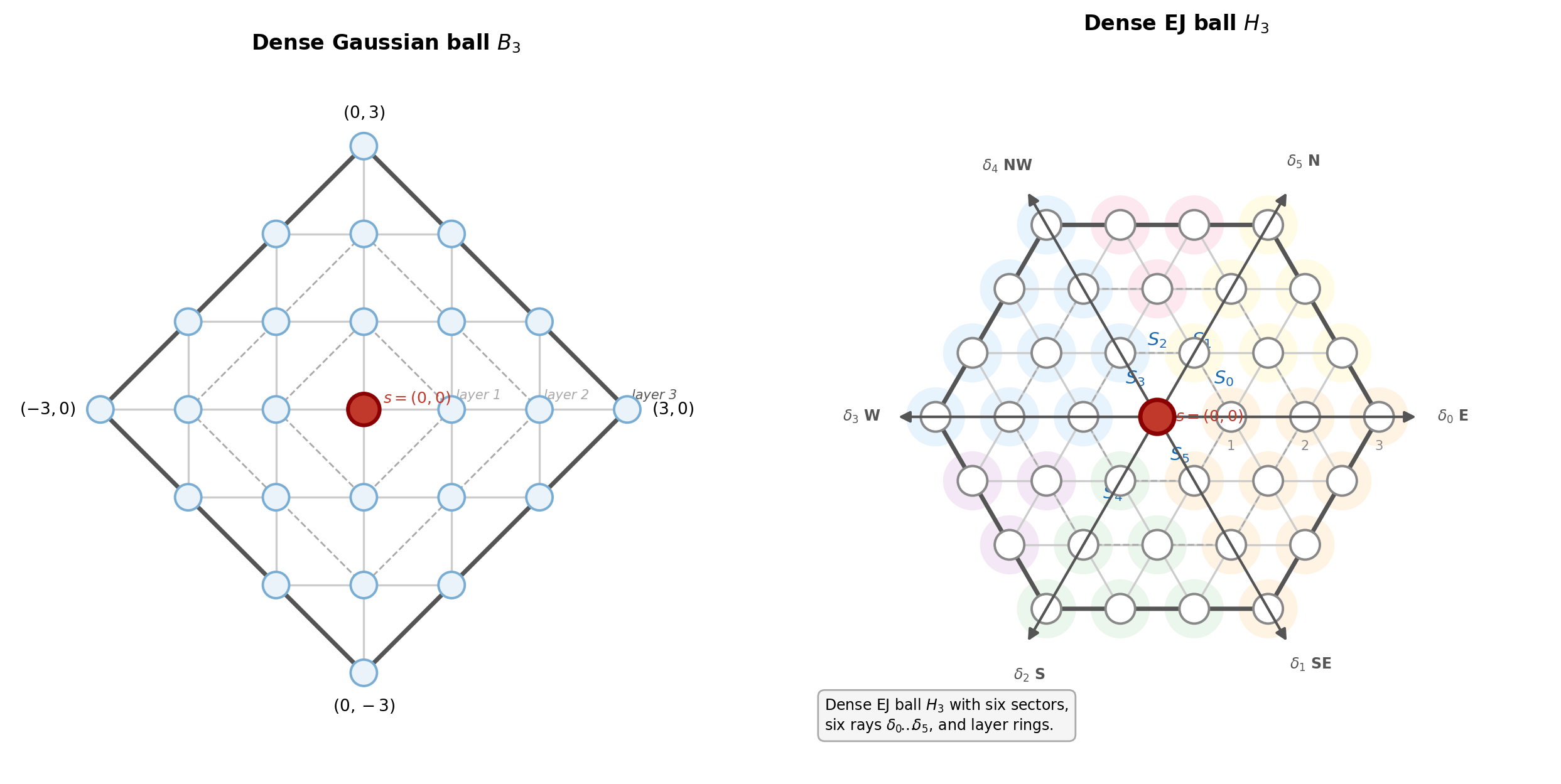}\hfill
\includegraphics[width=0.6\linewidth]{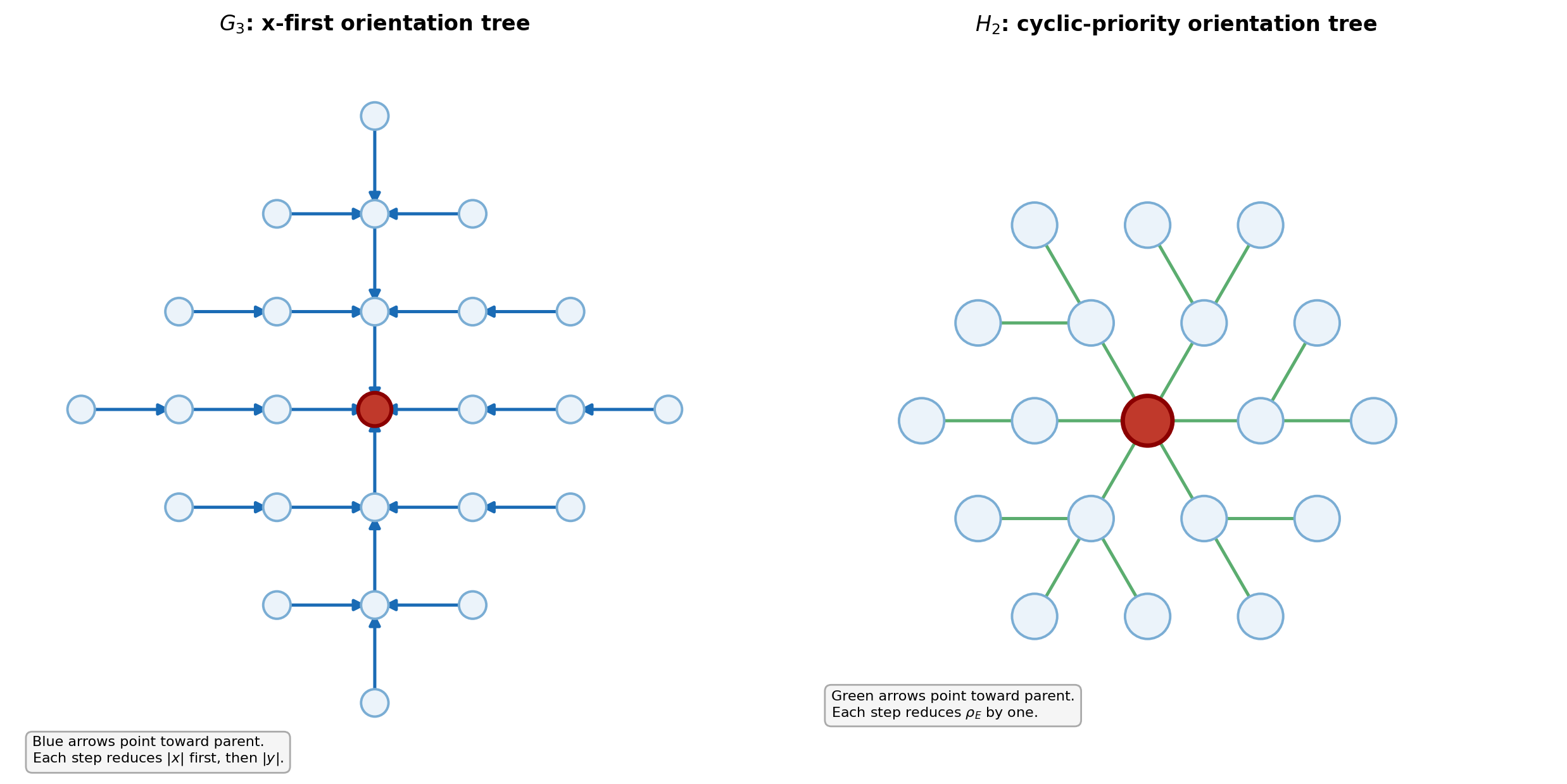}
\caption{Network geometry and fault-free broadcast trees. Left: the Gaussian ball $B_6$ and the EJ ball $H_3$. Right: schematic coordinate-reduction trees for $G_3$ and $H_2$. These figures visualize the layer sets and the spanning tree whose fault-pruned components the selector repairs.}
\label{fig:geometrytrees}
\end{figure}

\subsection{External Component Repair}

Let $T$ be a coordinate-reduction tree rooted at $s$, and let $F$ be a source-free fault set. The fault-pruned tree is
\begin{equation}
    T-F=C_1\cup C_2\cup\cdots\cup C_c,
\end{equation}
where $C_s$ denotes the component containing the source. A repair edge is external if its endpoints lie in two different components of $T-F$.

\begin{lemma}[Component lower bound]
\label{lem:lower}
Any repaired tree that preserves the $c$ healthy components of $T-F$ must use at least $c-1$ external component-crossing repair edges.
\end{lemma}
\begin{proof}
Contract every component of $T-F$ into a supernode. Any repaired broadcast tree over all healthy vertices induces a connected graph on these $c$ supernodes. A connected graph on $c$ vertices has at least $c-1$ edges, and each such edge corresponds to an external component-crossing edge in the original graph.
\end{proof}

\begin{definition}[Depth certificate]
Fix a fault-pruned orientation tree with components $C_1,\ldots,C_c$. A $K$-depth certificate is an ordering of the non-source components and crossing edges $(a_j,b_j)$ such that $b_j$ is in the next component, $a_j$ is in the already repaired part, and
\begin{equation}
    d(s,a_j)+1+\ecc_{C_j}(b_j)\le K.
    \label{eq:certineq}
\end{equation}
\end{definition}

\begin{figure}[H]
\centering
\includegraphics[width=0.65\linewidth]{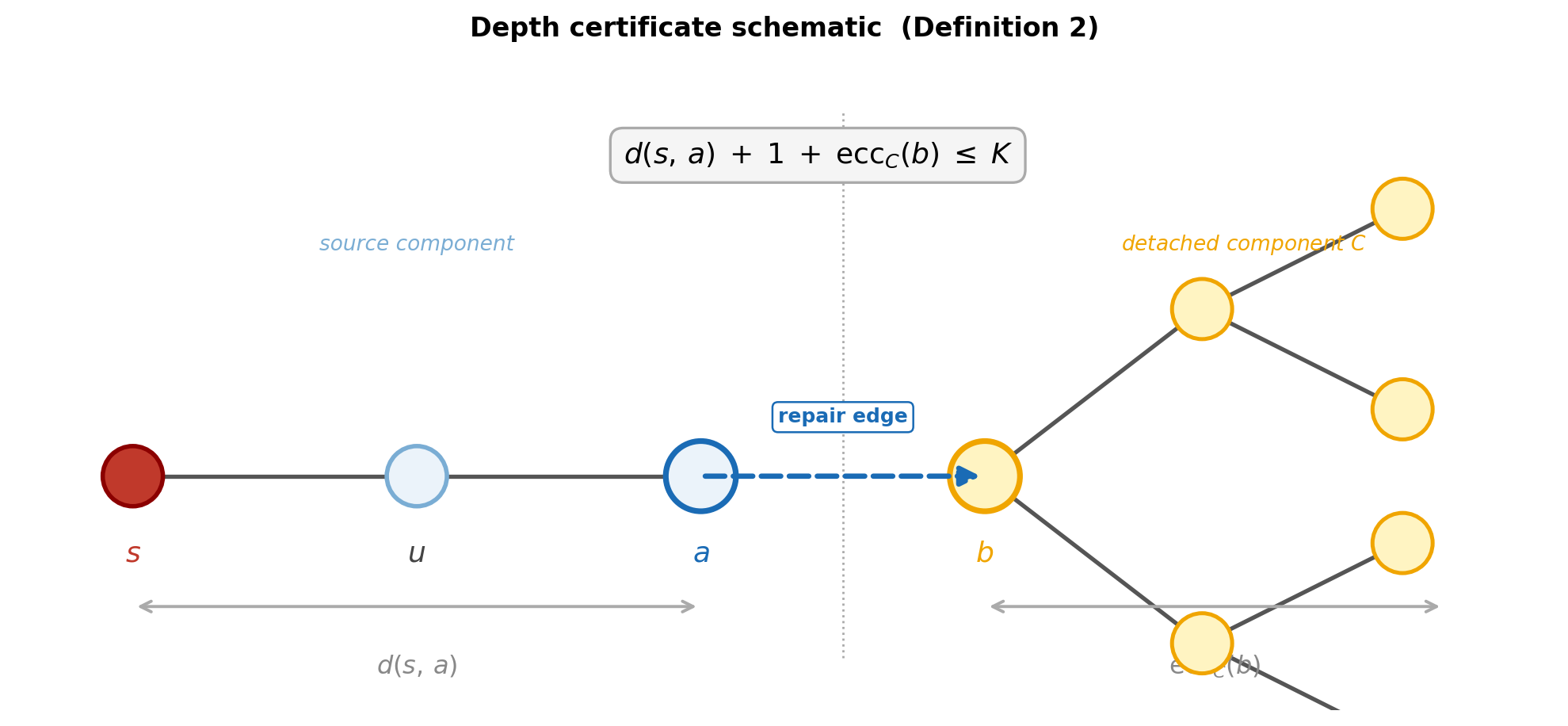}
\caption{Depth-certificate schematic for Definition~2 and Lemma~3. The already repaired part contains $s$ and $a$; the next detached component $C$ is entered at $b$. The depth contribution of attaching $C$ is bounded by $d(s,a)+1+\ecc_C(b)\le K$.}
\label{fig:depthcertificate}
\end{figure}

\begin{lemma}[Certificate lemma]
\label{lem:cert}
If $T-F$ admits a $K$-depth certificate using $c-1$ external edges, then adding those edges gives a non-redundant repaired broadcast tree of depth at most $K$ and with the minimum possible number of external repair edges for the selected orientation.
\end{lemma}
\begin{proof}
Each certificate edge attaches one previously unrepaired tree component to the already repaired tree. Adding one edge between a tree and a disjoint tree cannot create a cycle. After $c-1$ steps all $c$ components are connected, so the result is connected and acyclic over all healthy vertices. The depth of every vertex in a newly attached component is bounded by \eqref{eq:certineq}. Minimum repair count follows from Lemma~\ref{lem:lower}.
\end{proof}

\begin{figure}[H]
\centering
\includegraphics[width=0.75\linewidth]{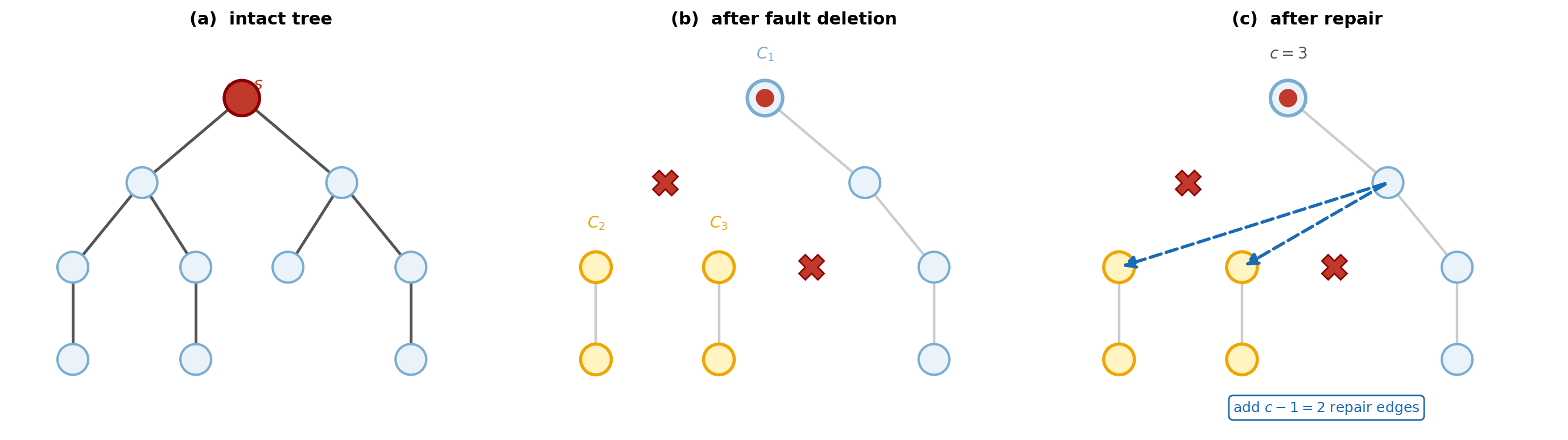}
\caption{Fault fragmentation and component repair. An intact broadcast tree becomes a fault-pruned forest with components $C_1,\ldots,C_c$; adding exactly $c-1$ external component-crossing repair edges reconnects these components into a rooted spanning tree of the healthy subgraph.}
\label{fig:repaircartoon}
\end{figure}

\section{Constant-Time Certificate Selectors}

A certificate selector receives the network parameter and a source-free fault set $F$ with $|F|\le2$. It returns a coordinate-reduction orientation and an ordered repair-edge list. The selector uses only a fixed number of integer comparisons, sign tests, sector-index computations, and modular adjacency tests. The complete repaired parent map can then be materialized by applying the selected orientation to all healthy vertices and replacing the parent relation at the selected component entries.

\begin{algorithm}[H]
\caption{Constant-Time Certificate Selection}
\label{alg:general}
\begin{algorithmic}[1]
\REQUIRE Family $\mathcal{N}\in\{G,EJ\}$, parameter $k$ or $t$, source-free fault set $F$, $|F|\le2$.
\ENSURE Coordinate-reduction orientation $\theta$ and ordered repair list $R$.
\IF{$\mathcal{N}=G$}
    \STATE Classify the faults by axis membership, ray relation, mixed axis/off-axis relation, orthogonal-axis relation, and off-axis row relation.
    \STATE Select the Gaussian certificate case in Table~\ref{tab:gcases}.
\ELSE
    \STATE Classify the faults by EJ ray membership, sector membership, same/opposite ray relation, ray/sector mixed relation, and adjacent/non-adjacent sector relation.
    \STATE Select the EJ certificate case in Table~\ref{tab:ejcases}.
\ENDIF
\STATE Return the orientation and the ordered repair edges specified by the selected certificate case.
\end{algorithmic}
\end{algorithm}

The selector is constant time because Tables~\ref{tab:gcases} and~\ref{tab:ejcases} have constant size and each entry uses a bounded number of coordinate expressions. Boundary wraparound is not treated geometrically by drawing the hexagon or diamond; it is treated algebraically by the congruence maps \eqref{eq:phiG} and \eqref{eq:phiE}. Thus an edge whose canonical coordinates wrap at the boundary is accepted exactly when its label difference is one of the network generators modulo the network order.

\begin{table}[H]
\centering
\caption{Symmetry normalization used by the selectors. The inverse map is applied to both the selected orientation and every selected repair edge.}
\label{tab:symnorm}
\begin{adjustbox}{max width=\textwidth}
\begin{tabular}{p{0.12\linewidth}p{0.38\linewidth}p{0.40\linewidth}}
\toprule
Family & Canonical representative & Inverse map after selection\\
\midrule
Gaussian & Reflect signs so the primary axis/ray is positive. When needed, exchange $(x,y)$ and $(y,x)$ so the selected formula is written in the $x$-priority representative. & Undo the sign reflection and, if used, undo the coordinate exchange. The inverse is applied to the orientation name and to each ordered repair edge before the plan is returned.\\
EJ & Apply a dihedral symmetry of the hexagon so the first relevant ray is $\delta_0$ and the interacting sector lies in the canonical cyclic position used in Table~\ref{tab:ejselector}. & Undo the dihedral permutation on ray indices, sector indices, the selected orientation family, and every coordinate endpoint of every repair edge.\\
\bottomrule
\end{tabular}
\end{adjustbox}
\end{table}

Table~\ref{tab:symnorm} is included because both selectors are stated on normalized representatives. The proofs are carried out only for those representatives; correctness for every placement follows because the network metrics, quotient-neighbor relations, and certificate inequalities are invariant under the corresponding symmetry group.

\subsection{Algebraic Selector Primitives}
\label{subsec:selector-primitives}

The selector tables use two small algebraic primitives. They are included to make the construction independent of drawings and to handle boundary wraparound uniformly.

For the Gaussian network define
\begin{equation}
\begin{aligned}
\Gamma_G(u)=\{v\in B_k:
&\ \phi_G(v)-\phi_G(u)\equiv \pm k \\
&\text{ or }\ \pm(k+1) \pmod {N_G}\}.
\end{aligned}
\end{equation}
For the EJ network define
\begin{equation}
\begin{aligned}
\Gamma_E(u)=\{v\in H_t:
&\ \phi_E(v)-\phi_E(u)\equiv \pm\phi_E(\delta_i) \\
&\pmod {N_E},\ 0\le i<6\}.
\end{aligned}
\end{equation}
Thus $\Gamma_G$ and $\Gamma_E$ are quotient-neighbor maps; if an ordinary coordinate neighbor exits the canonical ball, the congruence returns its unique canonical representative. This is the only boundary rule used in the algorithms.

For a finite ordered list $L=(e_1,\ldots,e_m)$ of possible edges and a fault set $F$, let
\begin{equation}
\begin{aligned}
\First_{\mathcal N}(L,F)=&\text{ the first }e_j=\{u,v\}\text{ such that}\\
&u,v\notin F\text{ and }v\in\Gamma_{\mathcal N}(u),
\end{aligned}
\end{equation}
where $\mathcal N$ is $G$ or $E$. If no such edge exists, the entry is omitted. Since all lists in Tables~\ref{tab:gselector} and~\ref{tab:ejselector} have bounded length, $\First_{\mathcal N}$ is a constant-time algebraic operation.

For a selected orientation $\theta$, the only components that can be detached by one or two faults are the children of the faulty vertices in $T_\theta$ that remain healthy. The selector tables give one ordered attachment formula for each possible child-root family. A listed edge is used exactly when its component is nonempty; empty intervals and leaf deletions therefore add no repair edge. The proofs below show that the listed formulas attach all non-source components and never attach two edges to the same component.

\noindent\textit{Family-specific execution.} After Algorithm~\ref{alg:general} chooses the network family, the Gaussian selector normalizes the fault set, reads the corresponding row of Table~\ref{tab:gselector}, evaluates the bounded lists by $\First_G$, omits empty child components, and maps the selected edges back by the inverse symmetry in Table~\ref{tab:symnorm}. The EJ selector performs the same steps with the dihedral normalization, the operators in Table~\ref{tab:ejselector}, and $\First_E$. This inline description avoids long floating pseudocode and keeps the table-to-proof mapping visible.

\section{Gaussian Certificates}

This section proves the Gaussian part. The source is $(0,0)$, and all coordinates are in $B_k$.

\subsection{Gaussian Case Partition}

A coordinate is an axis coordinate if $xy=0$, and off-axis otherwise. For two axis faults, there are three mutually exclusive cases: same ray, opposite ray, and orthogonal axes. For one axis and one off-axis fault we use the mixed case. For two off-axis faults we use the O3 off-axis case, with the same-row and different-row subcases included.

\begin{table}[H]
\centering
\caption{Normalized Gaussian selector formulas. Here $\epsilon,\eta\in\{\pm1\}$ are signs, $1\le a<b\le k$, and every row is reflected or coordinate-swapped back to the original position. Each $L_j$ is an ordered list passed to $\First_G$; empty detached components are omitted.}
\label{tab:gselector}
\begin{adjustbox}{max width=\textwidth}
\begin{tabular}{p{0.18\linewidth}p{0.17\linewidth}p{0.53\linewidth}}
\toprule
Normalized faults & Orientation & Ordered repair lists \\
\midrule
$\{(a,0)\}$ & $x$-first & $L_1=((a+1,-1),(a+1,0)),((a+1,1),(a+1,0))$ if $a<k$ \\
$\{(x,y)\},\ xy\ne0$ & $x$-first & $L_1=((x+\epsilon,y-\eta),(x+\epsilon,y)),((x-\epsilon,y+\eta),(x,y+\eta))$, where $\epsilon=\sgn x$, $\eta=\sgn y$ \\
$\{(a,0),(b,0)\}$ & $x$-first & $L_1$ for the interval after $a$ if $a+1<b$; $L_2$ for the tail after $b$ if $b<k$, both using the one-axis form \\
$\{(a,0),(-b,0)\}$ & $x$-first & One tail list on each side: $((a+1,\pm1),(a+1,0))$ and $((-b-1,\pm1),(-b-1,0))$ when nonempty \\
$\{(a,0),(0,b)\}$ & O6 & $e_1,e_2,e_3,e_4$ from \eqref{eq:o6e1}--\eqref{eq:o6e4}; invalid or empty-component entries are omitted \\
$\{(a,0),(a-1,-1)\}$ & $x$-first & Diagonal-adjacent mixed patch: use $((a+1,1),(a+1,0)),((a,1),(a,0)),((a,-2),(a,-1))$ before the generic mixed list \\
axis/off-axis, otherwise & axis-first & Axis-tail list plus the off-axis side-entry list, ordered by increasing cut layer \\
two off-axis & $x$-first if rows differ, else $y$-first & One side-entry list for each fault; same-row pairs use coordinate exchange before applying the side-entry lists \\
\bottomrule
\end{tabular}
\end{adjustbox}
\end{table}

\noindent\textit{How to read Table~\ref{tab:gselector}.} Each row gives a normalized representative; signs and coordinate swaps are applied before the row is used and inverted after the edge list is chosen. The lists are not witnesses found by search. They are algebraic formulas for the first possible component entry. The operator $\First_G$ chooses between boundary-equivalent transverse candidates and discards invalid boundary representatives. A row contributes an edge only when the corresponding interval, tail, suffix, or cap is nonempty. Examples~\ref{ex:gsameray} and~\ref{ex:goffrow} expand two rows that are otherwise easy to misread.

\begin{table}[H]
\centering
\caption{Gaussian constant-time certificate table for $|F|\le2$. Empty suffixes are ignored. Reflections and coordinate exchange are applied algebraically.}
\label{tab:gcases}
\begin{adjustbox}{max width=\textwidth}
\begin{tabular}{p{0.20\linewidth}p{0.25\linewidth}p{0.22\linewidth}p{0.12\linewidth}p{0.12\linewidth}}
\toprule
Fault signature & Selected orientation & Component structure & Max $c$ & Depth bound\\
\midrule
One axis fault & Axis-priority tree with transverse boundary entry & One axis tail or boundary micro-tail & $2$ & $k+2$\\
One off-axis fault & $x$-priority or reflected $x$-priority side-entry tree & One monotone suffix & $2$ & $k$\\
Two axis faults, same ray & Same-axis priority tree & Bounded interval plus outer tail & $3$ & $k+2$\\
Two axis faults, opposite rays & Same-axis priority tree with two transverse entries & Two independent tails & $3$ & $k+2$\\
Mixed axis/off-axis & Axis-priority tree; diagonal-adjacent subcase uses near-corner entry & Axis tail plus one off-axis suffix & $3$ & $k+2$\\
Orthogonal-axis pair & Reflected O6 tree & Axis tail, two row fragments, upper cap & $5$ & $k+2$\\
Two off-axis faults & O3 side-entry tree; same-row pairs use coordinate exchange & One or two off-axis suffixes & $3$ & $k+1$\\
\bottomrule
\end{tabular}
\end{adjustbox}
\end{table}

\subsection{Layer-Suffix Inequality}

\begin{lemma}[Gaussian suffix inequality]
\label{lem:gsuffix}
Let $C$ be a detached Gaussian suffix whose entry $b$ is beyond a cut at layer $r$. If $\ecc_C(b)\le k-r$ and the repaired-side endpoint $a$ has repaired depth at most $r+1$, then attaching $C$ through $(a,b)$ has depth contribution at most $k+2$. If $\ecc_C(b)\le k-r-1$ and $d(s,a)\le r$, then the contribution is at most $k$.
\end{lemma}
\begin{proof}
The first bound follows from
\begin{equation}
    d(s,a)+1+\ecc_C(b)\le (r+1)+1+(k-r)=k+2.
\end{equation}
The second follows from
\begin{equation}
    d(s,a)+1+\ecc_C(b)\le r+1+(k-r-1)=k.
\end{equation}
\end{proof}

\begin{lemma}[Gaussian axis cases]
\label{lem:gaxis}
Every one-axis, same-ray two-axis, and opposite-ray two-axis fault placement admits a selected certificate satisfying the bounds in Table~\ref{tab:gcases}.
\end{lemma}
\begin{proof}
It suffices to analyze the $x$-axis; the $y$-axis case is obtained by exchanging the two coordinate generators. Reflections preserve $\rho_G$, the set of Gaussian generator edges, and the inequalities in Lemma~\ref{lem:gsuffix}.

For a one-fault placement $F=\{(a,0)\}$ with $a>0$, the selected $x$-axis priority tree can detach only the outer tail
\begin{equation}
        T_a=\{(u,0):a<u\le k\},
\end{equation}
with the convention that $T_a=\emptyset$ when $a=k$. If $T_a\ne\emptyset$, its first vertex is $b=(a+1,0)$. One of the two transverse candidates $(a+1,1)$ and $(a+1,-1)$ is a healthy quotient-neighbor of $b$; at boundary layers this is interpreted through $\Gamma_G$. The repaired-side endpoint has depth at most $a+2$ in the selected tree and the tail eccentricity from $b$ is at most $k-a-1$. Hence the attachment value is at most
\begin{equation}
        (a+2)+1+(k-a-1)=k+2.
\end{equation}
If the tail is empty, $c=1$ and no repair edge is used.

For same-ray faults $(a,0)$ and $(b,0)$ with $1\le a<b\le k$, the selected tree is cut into at most two non-source pieces: the bounded interval
\begin{equation}
        I_{a,b}=\{(u,0):a<u<b\}
\end{equation}
if it is nonempty, and the outer tail $T_b$ if $b<k$. The interval entry has repaired-side depth at most $a+2$ and eccentricity at most $b-a-2$, giving value at most $b+1\le k+1$. The tail entry has repaired-side depth at most $b+2$ and eccentricity at most $k-b-1$, giving value at most $k+2$. The two entries lie on disjoint intervals, so each selected edge attaches a different component. Thus $c\le3$ and the number of useful repair edges is exactly $c-1$.

For opposite-ray faults $(a,0)$ and $(-b,0)$, $a,b>0$, there are at most two nonempty tails, one on each side of the source. We verify explicitly that the two transverse entries do not block each other. The positive-side tail starts at $(a+1,0)$ and uses a transverse endpoint of the form $(a+1,\pm1)$ or its quotient representative. This endpoint is not the negative-ray fault $(-b,0)$ because its first coordinate is positive and its second coordinate is nonzero before quotient normalization; under $\Gamma_G$ its label differs from $(a+1,0)$ by one generator, whereas the negative-ray fault is separated by at least the source layer. Moreover, in the selected $x$-priority tree the parent chain from $(a+1,\pm1)$ first reduces the nonzero transverse coordinate only after the $x$-coordinate has been reduced to the source side; it cannot pass through $(-b,0)$, which is on the opposite ray. The reflected argument applies to the negative-side tail and the positive fault. Therefore the two tails have independent healthy entries. Each value is bounded by $(a+2)+1+(k-a-1)\le k+2$ or by the reflected $b$-analogue. Hence $c\le3$ and the repair uses exactly $c-1\le2$ edges.
\end{proof}

\begin{lemma}[Gaussian off-axis and mixed cases]
\label{lem:goffmixed}
Every one-off-axis, mixed axis/off-axis, and two-off-axis placement admits a selected certificate satisfying Table~\ref{tab:gcases}.
\end{lemma}
\begin{proof}
For one off-axis fault $f=(x,y)$, $xy\ne0$, write $r=\rho_G(f)$ and let $\epsilon=\sgn(x)$, $\eta=\sgn(y)$. The side-entry list in Table~\ref{tab:gselector} contains one candidate on each transverse side of the descendant suffix. The selected endpoint $a$ outside the suffix has layer at most $r$, and the component entry $b$ has remaining eccentricity at most $k-r-1$. Lemma~\ref{lem:gsuffix} gives value at most $k$. Since a single off-axis deletion can detach only one monotone suffix in the selected orientation, $c\le2$.

Consider a mixed placement consisting of an axis fault
and an off-axis fault. By reflection and coordinate exchange,
assume the axis fault is $(a,0)$ with $a>0$ and the off-axis
fault is $f=(x,y)$, $xy\ne0$. We first identify precisely which
placements allow blocking. In the $x$-priority orientation, the
off-axis suffix $C_f$ is monotone: its first vertex beyond the cut
is
\begin{equation}
    f'=(x+\epsilon,y),
\end{equation}
and the two side-entry endpoints are
\begin{equation}
    e^-=(x+\epsilon,y-\eta),\qquad
    e^+=(x-\epsilon,y+\eta),
\end{equation}
where $\epsilon=\operatorname{sgn}(x)$ and
$\eta=\operatorname{sgn}(y)$. The axis fault $(a,0)$ blocks a
side-entry endpoint $e$ either by occupying $e$ directly or by
lying on the $x$-priority parent chain from $e$ to the source
before that chain reaches the source. In the $x$-priority tree,
the parent chain from $e$ reduces $|x|$ at every step until
$x=0$, then reduces $|y|$; it therefore visits the $x$-axis only
at the single node $(0,0)$ unless $e$ itself lies on the $x$-axis.
Consequently, $(a,0)$ can lie on the parent chain of $e$ only
if $e$ is on the $x$-axis, which requires $y-\eta=0$ for $e^-$
or $y+\eta=0$ for $e^+$, i.e., $|y|=1$. For $e$ to be the fault
$(a,0)$ directly, the first coordinate must satisfy $x+\epsilon=a$
for $e^-$ or $x-\epsilon=a$ for $e^+$, and the second coordinate
must be zero, again forcing $|y|=1$. In both cases $|y|=1$ and
the first coordinate of the blocked endpoint is $a$ or $a\pm1$;
under the positive-$x$ normalization this reduces to
$x\in\{a-1,a\}$. After reflecting the sign of $y$ to $-1$ and
applying coordinate exchange if needed, the only blocked family is
\begin{equation}
F = \{(a,0),(a-1,-1)\}
\end{equation}
under reflection and coordinate exchange. These are exactly
the rows for which the selector invokes the diagonal-adjacent
mixed patch in Table~\ref{tab:gselector}.

If the placement is not diagonal-adjacent, the off-axis side entry is not the axis fault and its parent chain does not pass through the axis fault. The axis tail and the off-axis suffix are therefore attached independently, in increasing cut layer. The axis piece has value at most $k+2$ by Lemma~\ref{lem:gaxis}, and the off-axis suffix has value at most $k$ by the preceding paragraph. If the placement is diagonal-adjacent, the patch selects the opposite transverse side. For the representative $\{(a,0),(a-1,-1)\}$, the repaired-side endpoint of the patched off-axis entry has depth at most $a+1$, while the remaining suffix eccentricity is at most $k-a$. Hence the value is at most
\begin{equation}
        (a+1)+1+(k-a)=k+2.
\end{equation}
The axis component, if nonempty, is attached by its ordinary transverse tail edge. Thus every mixed placement has $c\le3$ and uses exactly $c-1$ repair edges.

It remains to consider two off-axis faults. If their row coordinates differ, the two side-entry endpoints lie on distinct transverse rows after normalization. One side entry can be blocked only by a fault on the same transverse row, which is excluded; hence at least one entry for each suffix remains healthy. Each non-boundary suffix has value at most $k$ by Lemma~\ref{lem:gsuffix}. A boundary microcomponent has one or two vertices, so its eccentricity is at most one and it is attached from a source-component boundary endpoint of depth at most $k$, giving value at most $k+1$.

If the two off-axis faults are on the same row, the selector applies either reflection in that row or the coordinate exchange $(x,y)\mapsto(y,x)$ before the side-entry rule. This exchange preserves $\rho_G$, sends unit coordinate edges to unit coordinate edges, and swaps the generator differences $k$ and $k+1$, so the quotient adjacency and all certificate values are unchanged. After the exchange, the pair is either a different-row pair, already handled, or an ancestor-ordered pair on a monotone suffix. In the ancestor case the inner detached component is attached first; its eccentricity terminates before the outer cut and its value is at most $k$. The outer suffix is then attached through the remaining side entry and has value at most $k+1$, including the boundary microcomponent case. Thus $c\le3$ and the two-off-axis depth is at most $k+1\le k+2$.
\end{proof}

\begin{lemma}[Gaussian O6 orthogonal-axis certificate]
\label{lem:o6}
Let, up to reflection and coordinate exchange,
\begin{equation}
    F=\{(a,0),(0,b)\},\qquad 1\le a,b,
\end{equation}
with both faults in $B_k$. The selected O6 orientation has at most five fault-pruned components and admits a $k+2$ certificate with at most four repair edges.
\end{lemma}
\begin{proof}
The selected orientation is the reflected O6 coordinate-priority tree whose source component contains the lower half-ball. Deletion of $(a,0)$ and $(0,b)$ separates at most four non-source pieces: the positive $x$-tail, two row fragments at height $b$, and the upper cap. Empty pieces are ignored. For the normalized positive representative, the certificate uses the following ordered edges when their target components are nonempty:
\begin{align}
    e_1&=\{(-a,b-1),(-a,b)\},\label{eq:o6e1}\\
    e_2&=\{(1,b-1),(1,b)\},\label{eq:o6e2}\\
    e_3&=\{(a+1,-1),(a+1,0)\},\label{eq:o6e3}\\
    e_4&=\{(1,b),(1,b+1)\}.\label{eq:o6e4}
\end{align}
For $e_1$, the repaired-side endpoint has depth $a+b-1$, and the row-fragment eccentricity is at most $k-b-a$, giving value at most $k$. For $e_2$, the repaired-side endpoint has depth $b$, and the right row-fragment eccentricity is at most $k-b-1$, giving value at most $k$. For $e_3$, the repaired-side endpoint has depth $a+2$, and the axis-tail eccentricity is at most $k-a-1$, giving value at most $k+2$. After $e_2$ attaches the right row fragment, $e_4$ attaches the upper cap from a vertex of repaired depth at most $b+1$, with cap eccentricity at most $k-b$, giving value at most $k+2$. The four edges attach the four possible non-source components, so the repair uses exactly $c-1$ edges.
\end{proof}

\begin{theorem}[Gaussian constant-time repair]
\label{thm:gaussian}
For every $k\ge5$ and every source-free fault set $F\subseteq B_k$ with $|F|\le2$, the Gaussian selector returns a coordinate-reduction orientation and exactly $c-1$ external component-crossing repair edges for that selected orientation. The repaired tree is non-redundant and has depth at most $k+2$.
\end{theorem}
\begin{proof}
The cases in Table~\ref{tab:gcases} are mutually exclusive and exhaustive: a single fault is axis or off-axis; two faults are either both axis, one axis and one off-axis, or both off-axis; the both-axis case further splits into same ray, opposite ray, and orthogonal axes. Lemmas~\ref{lem:gaxis}, \ref{lem:goffmixed}, and \ref{lem:o6} provide the required certificates. Lemma~\ref{lem:cert} converts each certificate into a non-redundant repaired tree and combines the depth bound with exact $c-1$ repair-edge optimality.
\end{proof}

\begin{figure}[H]
\centering
\includegraphics[width=0.6\linewidth]{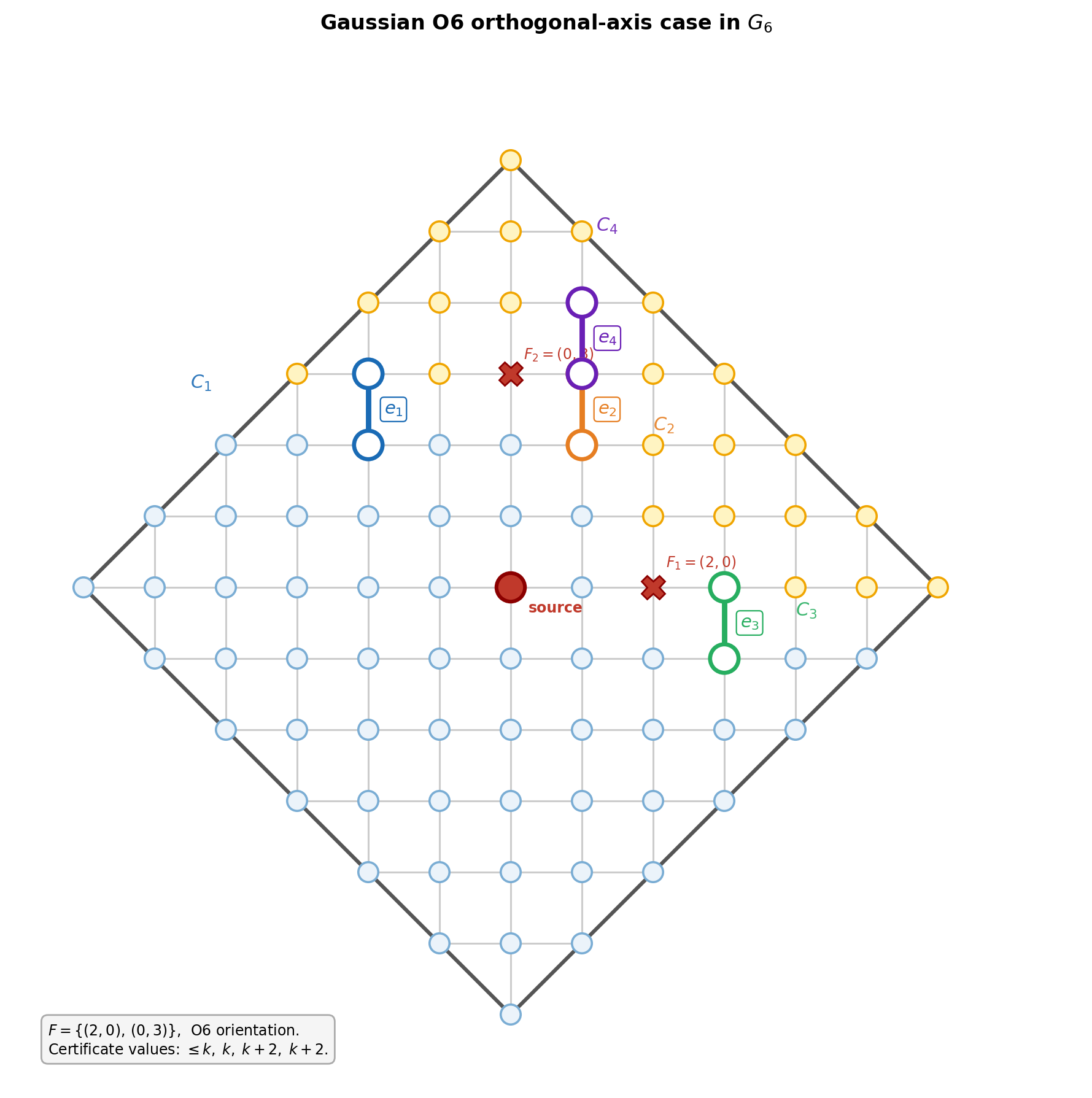}
\caption{Gaussian O6 orthogonal-axis case in $G_6$ with faults $F=\{(2,0),(0,3)\}$. The figure highlights the four repair edges $e_1$--$e_4$ from \eqref{eq:o6e1}--\eqref{eq:o6e4} and the resulting five-component structure.}
\label{fig:o6}
\end{figure}

\subsection{Gaussian Examples}

The following examples unpack the table entries rather than serving as proof. They show how a selected row determines the component order, the repaired-side depths, and the eccentricity terms in the certificate inequality.

\begin{example}[Orthogonal-axis repair in $G_6$]
Let $k=6$ and $F=\{(2,0),(0,3)\}$. Figure~\ref{fig:o6} shows the selected O6 orientation. The certificate edges are
\begin{align*}
    e_1&=\{(-2,2),(-2,3)\},& e_2&=\{(1,2),(1,3)\},\\
    e_3&=\{(3,-1),(3,0)\},& e_4&=\{(1,3),(1,4)\}.
\end{align*}
The first two edges attach row fragments, $e_3$ attaches the positive $x$-tail, and $e_4$ attaches the upper cap after the right row fragment is present. The certificate values are at most $6,6,8,8$, respectively; hence the final depth is at most $k+2=8$.
\end{example}

\begin{example}[Mixed diagonal-adjacent repair]
For $k=7$ and $F=\{(3,0),(2,-1)\}$, the off-axis suffix begins next to the axis fault. A naive lower-side entry reaches the suffix only after the path has moved around the axis cut, increasing the repaired-side depth by two layers. The diagonal patch uses the opposite transverse side. Its repaired-side endpoint has layer at most $4$, and the remaining suffix height is at most $4$, so the selected value is $4+1+4=9=k+2$.
\end{example}

\begin{example}[Same-ray Gaussian repair]
\label{ex:gsameray}
Let $k=5$ and $F=\{(1,0),(3,0)\}$. Table~\ref{tab:gselector} selects the same-ray row. The bounded interval component is the single vertex $(2,0)$, so the first list is nonempty. The outer tail begins at $(4,0)$ and contributes the second list. The interval entry has repaired-side depth at most $3$ and eccentricity $0$, while the tail entry has repaired-side depth at most $5$ and eccentricity $1$; both values are at most $7=k+2$. If the two faults were consecutive, the bounded interval would be empty and the first list would be omitted.
\end{example}

\begin{example}[Two off-axis same-row faults]
\label{ex:goffrow}
Take $k=6$ and $F=\{(2,1),(-1,1)\}$. The selector first applies the coordinate exchange $(x,y)\mapsto(y,x)$, producing $F'=\{(1,2),(1,-1)\}$. In the exchanged coordinates the pair is handled by the ordinary off-axis side-entry rule. The exchange preserves $\rho_G$ and maps generator edges to generator edges, so the two certificate values are unchanged after the selected edges are mapped back. The resulting depth is at most $k+1=7$.
\end{example}

\begin{example}[Boundary wraparound edge]
Let $k=5$ and $u=(5,0)$. The ordinary east neighbor of $u$ leaves the displayed diamond, but $\Gamma_G$ returns the unique $v\in B_5$ satisfying $\phi_G(v)-\phi_G(u)\equiv k\pmod{N_G}$. Thus a boundary repair edge is accepted by the same algebraic generator test as an interior edge.
\end{example}

\section{EJ Certificates}

The EJ proof uses the six cyclic directions $\delta_0,\ldots,\delta_5$. A vertex is ray-resident when it lies on one of the six boundary rays from the source; otherwise it is sector-interior. Sectors are indexed cyclically between adjacent rays.

\begin{figure}[H]
\centering
\includegraphics[width=0.6\linewidth]{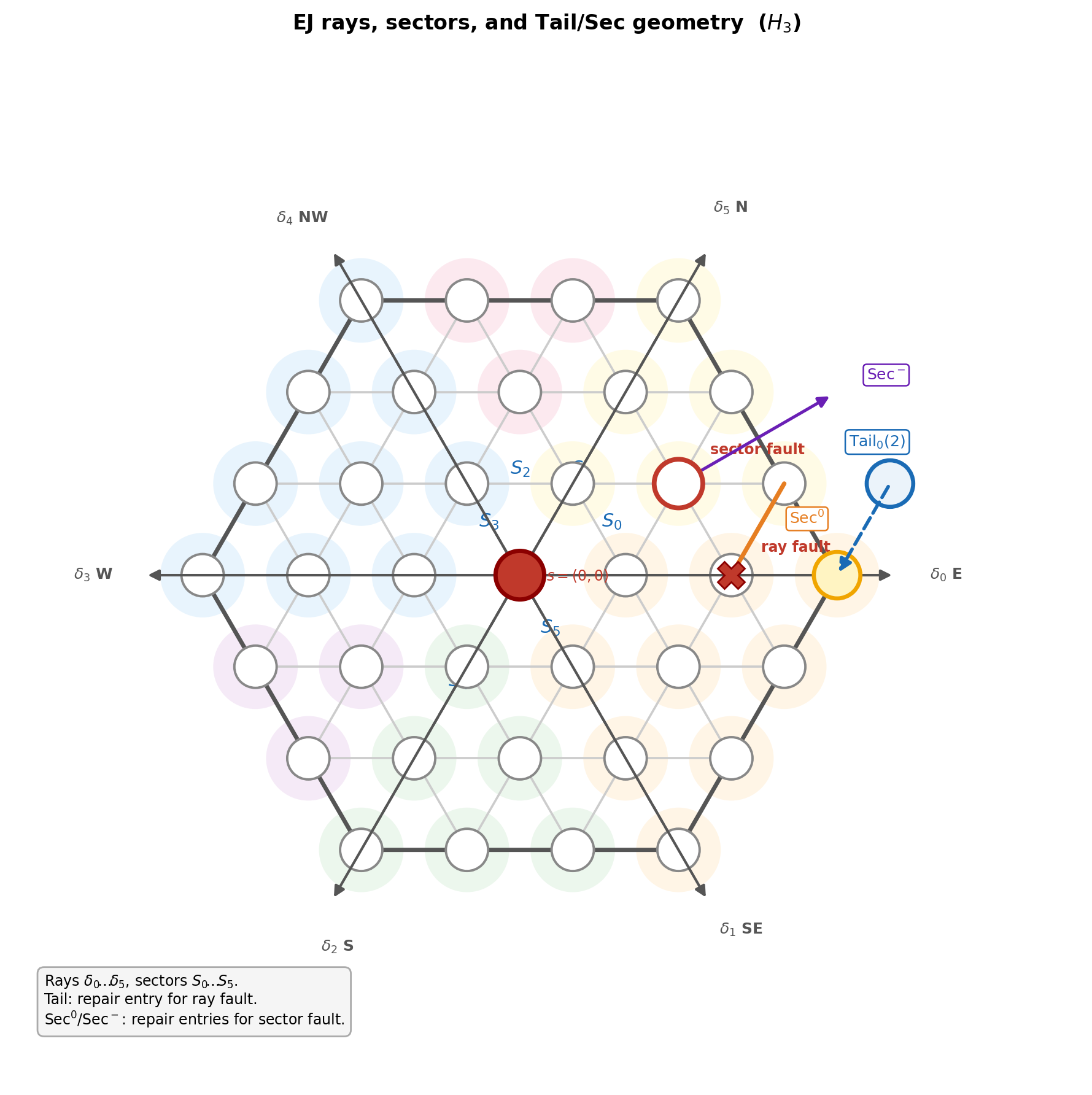}
\caption{EJ ray and sector geometry. The six rays $\delta_0,\ldots,\delta_5$ and six sectors $S_0,\ldots,S_5$ are shown on a small hexagonal ball. The marked examples illustrate a ray fault, a sector-interior fault, and the boundary-side attachments represented by the $\mathsf{Tail}$ and $\mathsf{Sec}$ operators.}
\label{fig:ejgeometry}
\end{figure}

\subsection{Explicit EJ Selector Operators}
\label{subsec:ej-operators}

The EJ selector can be written compactly using three constant-size edge-list operators. Indices are always modulo six. For a ray vertex $a\delta_i$, define the first healthy outward ray point
\begin{equation}
    q_i(a)=(a+1)\delta_i.
\end{equation}
If $a<t$, the ray-tail list is
\begin{equation}
\label{eq:ejraytail}
    \mathsf{Tail}_i(a)=\big(\{q_i(a)+\delta_{i-1},q_i(a)\},\{q_i(a)+\delta_{i+1},q_i(a)\}\big),
\end{equation}
and $\mathsf{Tail}_i(a)=\emptyset$ when $a=t$. The two possible repaired-side endpoints are the two transverse neighbors of the first tail vertex. The operator $\First_E$ selects the first healthy quotient-valid edge in the displayed order.

For a sector-interior vertex $f=u\delta_i+v\delta_{i+1}$ with $u,v>0$, define
\begin{align}
\label{eq:ejsectorop}
    \mathsf{Sec}^{0}_i(f)&=\{f+\delta_i, f+\delta_{i+1}\},\\
    \mathsf{Sec}^{-}_i(f)&=\{f+\delta_i, f+\delta_{i-1}\},\\
    \mathsf{Sec}^{+}_i(f)&=\{f+\delta_{i+1}, f+\delta_{i+2}\}.
\end{align}
The edge $\mathsf{Sec}^{0}_i(f)$ connects the two sibling first-layer vertices beyond the cut; the other two edges are the outer-side boundary alternatives. Let $h\in B_{i+1}$ mean that the other fault lies on the $(i+1)$-side boundary of the sector, and let $h\in B_i$ denote the reflected $i$-side condition. The ordered sector list is the following constant-size rule:
\begin{equation}
\label{eq:ejseclist}
\begin{array}{ll}
\mathsf{Sec}_i(f;h)=(\mathsf{Sec}^{-}_i,\mathsf{Sec}^{0}_i,\mathsf{Sec}^{+}_i), & h\in B_{i+1},\\
\mathsf{Sec}_i(f;h)=(\mathsf{Sec}^{+}_i,\mathsf{Sec}^{0}_i,\mathsf{Sec}^{-}_i), & h\in B_i,\\
\mathsf{Sec}_i(f;h)=(\mathsf{Sec}^{0}_i,\mathsf{Sec}^{-}_i,\mathsf{Sec}^{+}_i), & \text{otherwise},
\end{array}
\end{equation}
where the argument $(f)$ on each $\mathsf{Sec}$ entry is suppressed in \eqref{eq:ejseclist}. The first two lines choose a side away from a blocking fault; the third line is the one-fault or unblocked case. Each list has length at most three and is evaluated only by $\First_E$.

For two ray faults on the same ray $a\delta_i,b\delta_i$ with $a<b$, define
\begin{equation}
\label{eq:ejraypair}
    \mathsf{RayPair}_i(a,b)=\big(\mathsf{Tail}_i(a)\text{ if }a+1<b,\; \mathsf{Tail}_i(b)\text{ if }b<t\big).
\end{equation}
For opposite rays $a\delta_i$ and $b\delta_{i+3}$, use $\mathsf{Tail}_i(a)$ and $\mathsf{Tail}_{i+3}(b)$. For other ray pairs, sort the two ray indices cyclically and apply the same tail operator to each nonempty interval or outer tail; because only two faults are present, at most three such lists are generated.

\begin{table}[H]
\centering
\caption{Explicit normalized EJ selector formulas. Directions are $\delta_0=E,\delta_1=SE,\delta_2=S,\delta_3=W,\delta_4=NW,\delta_5=N$. Every row is mapped back by the inverse dihedral symmetry. Each list is evaluated by $\First_E$ and is omitted when the corresponding component is empty.}
\label{tab:ejselector}
\begin{adjustbox}{max width=\textwidth}
\begin{tabular}{p{0.17\linewidth}p{0.19\linewidth}p{0.52\linewidth}}
\toprule
Normalized faults & Orientation family & Ordered repair lists \\
\midrule
one ray fault $a\delta_i$ & $C_i$ or $R_i$ exposing a transverse side & $\mathsf{Tail}_i(a)$ from \eqref{eq:ejraytail} \\
one sector fault $f\in S_i$ & $C_i$ or $R_{i+1}$ & $\mathsf{Sec}_i(f;\emptyset)$ from \eqref{eq:ejseclist} \\
two same-ray faults $a\delta_i,b\delta_i$ & $A_i$ & $\mathsf{RayPair}_i(a,b)$ from \eqref{eq:ejraypair} \\
two opposite-ray faults $a\delta_i,b\delta_{i+3}$ & $A_i$ & $\mathsf{Tail}_i(a)$, then $\mathsf{Tail}_{i+3}(b)$, ordered by smaller layer \\
two other ray faults & cyclic priority on the smaller ray arc & Tail/interval lists obtained by splitting the cyclic ray arc at the two cuts; at most three $\mathsf{Tail}$ lists \\
ray/sector mixed $a\delta_j,f\in S_i$ & cyclic/reverse priority exposing the side of $S_i$ away from $\delta_j$ & $\mathsf{Tail}_j(a)$ and $\mathsf{Sec}_i(f;a\delta_j)$, ordered by increasing cut layer \\
two non-adjacent sector faults $f\in S_i,h\in S_j$ & cyclic/reverse priority exposing disjoint boundaries & $\mathsf{Sec}_i(f;h)$ and $\mathsf{Sec}_j(h;f)$; attach the suffix whose chosen boundary is not incident with the other sector first \\
two adjacent or same-sector faults $f,h$ & alternating priority around their shared boundary & outer-side $\mathsf{Sec}$ list first; if it is blocked, attach the unblocked suffix and then the shared-boundary sibling edge $\mathsf{Sec}^{0}$ \\
\bottomrule
\end{tabular}
\end{adjustbox}
\end{table}

\noindent\textit{How to read Table~\ref{tab:ejselector}.} The EJ rows use the same convention as the Gaussian rows: a normalized representative is selected by a dihedral symmetry, the ordered edge lists are evaluated by $\First_E$, and the inverse symmetry maps the chosen edges back. The operators $\mathsf{Tail}$, $\mathsf{RayPair}$, and $\mathsf{Sec}$ encode entire ray or sector fragments, but only one edge is selected for each nonempty fragment. Hence the selector returns exact component attachments, not a pool of alternatives. Figure~\ref{fig:ejgeometry} is the visual dictionary for these operators; Examples~\ref{ex:ejotherray} and~\ref{ex:ejadjsector} expand the two most compressed EJ rows.

\begin{table}[H]
\centering
\caption{EJ constant-time certificate table for $|F|\le2$. Cases are listed under the priority order used by the selector.}
\label{tab:ejcases}
\begin{adjustbox}{max width=\textwidth}
\begin{tabular}{p{0.21\linewidth}p{0.25\linewidth}p{0.23\linewidth}p{0.11\linewidth}p{0.11\linewidth}}
\toprule
Fault signature & Selected orientation family & Component structure & Max $c$ & Depth bound\\
\midrule
One ray fault & Cyclic/reverse priority exposing a transverse entry & One ray tail & $2$ & $t+1$\\
One sector fault & Sector side-entry priority & One sector suffix or none & $2$ & $t+1$\\
Two same-ray faults & Alternating opposite-pair priority & Bounded ray interval plus outer tail & $3$ & $t+2$\\
Two opposite-ray faults & Alternating priority separating two tails & Two independent tails & $3$ & $t+2$\\
Two other ray-pair faults & Cyclic priority selected by ray order & At most three ray fragments & $4$ & $t+2$\\
Ray/sector mixed & Ray-tail entry plus sector side-entry & Tail plus sector suffix & $3$ & $t+2$\\
Two sector faults, non-adjacent & Cyclic/reverse priority exposing separated boundaries & Two separated sector suffixes & $3$ & $t+2$\\
Two sector faults, adjacent or same & Alternating priority around shared boundary & Two sector suffixes or shared cap & $3$ & $t+2$\\
\bottomrule
\end{tabular}
\end{adjustbox}
\end{table}

\subsection{EJ Suffix Inequality}

\begin{lemma}[EJ suffix inequality]
\label{lem:ejsuffix}
Let $C$ be an EJ sector or ray suffix cut at layer $r$. If the component entry $b$ has $\ecc_C(b)\le t-r-1$ and the repaired-side endpoint $a$ has repaired depth at most $r+1$, then the attachment has value at most $t+1$. If $d(s,a)\le r+2$, then the attachment has value at most $t+2$.
\end{lemma}
\begin{proof}
The inequalities are
\begin{align}
    (r+1)+1+(t-r-1)&=t+1,\\
    (r+2)+1+(t-r-1)&=t+2.
\end{align}
\end{proof}

\begin{lemma}[EJ one-fault cases]
\label{lem:ejone}
Every one-fault placement in $H_t$ admits a selected certificate of depth at most $t+1$ and repair count $c-1\le1$.
\end{lemma}
\begin{proof}
If the fault lies on a ray at layer $r$, the selected cyclic or reversed priority exposes a transverse entry to the possible ray tail. The entry endpoint has depth at most $r+1$ and the remaining tail eccentricity is at most $t-r-1$, so Lemma~\ref{lem:ejsuffix} gives $t+1$. If the fault is sector-interior, the selected priority exposes one of the two sector boundaries as an entry. The descendant region is a sector suffix with the same eccentricity bound. If the deletion does not detach a non-source component, then $c=1$ and no repair edge is needed; otherwise $c=2$ and one edge suffices.
\end{proof}

\begin{lemma}[EJ ray-pair cases]
\label{lem:ejray}
Every two-fault pair in which both faults are ray-resident admits a selected $t+2$ certificate with at most three repair edges.
\end{lemma}
\begin{proof}
Let the six rays be ordered cyclically. Dihedral symmetry lets us normalize one fault to a positive layer on $\delta_i$.

If both faults lie on the same ray, order their layers $a<b$. The selected alternating priority exposes the two transverse neighbors of the ray before continuing along the damaged trunk. The bounded interval between the cuts, when nonempty, has entry depth at most $a+1$ and eccentricity at most $b-a-1$, giving value at most $b+1\le t+1$. The outer tail has entry depth at most $b+1$ and eccentricity at most $t-b-1$ unless it is a boundary microtail; the ordinary value is at most $t+1$, and the boundary microtail has eccentricity at most one and therefore value at most $t+2$. Thus $c\le3$.

If the faults lie on opposite rays, the same alternating orientation gives two tails whose transverse entries lie in disjoint half-hexagons. A transverse endpoint of the $\delta_i$ tail has a nonzero coefficient in one of the adjacent directions $\delta_{i-1}$ or $\delta_{i+1}$, whereas the opposite fault lies on the pure ray $\delta_{i+3}$. Hence the endpoint cannot equal the opposite fault. Its parent chain under the selected alternating priority reduces the transverse coefficient before it can cross the source into the opposite ray, so the chain cannot be blocked by the opposite-ray fault. The reflected statement holds for the other tail. Applying Lemma~\ref{lem:ejsuffix} to the two tails gives value at most $t+2$ and $c\le3$.

Now suppose the ray faults lie on two distinct non-opposite rays. The cyclic order of the two cut rays partitions the selected ray skeleton into at most three fragments: the interval on the shorter cyclic arc between the two cut layers, and up to two outer tails outside that interval. No fourth fragment can appear because two deleted ray vertices create at most three intervals in a one-dimensional cyclically ordered ray skeleton. For each nonempty fragment, the selector uses a transverse entry produced by the corresponding $\mathsf{Tail}$ list. The entry of the third fragment, when it exists, is not blocked by either fault: it contains a nonzero coefficient in an adjacent sector direction, while both faults are pure ray vertices, and its parent chain stays inside the sector incident with that fragment until it reaches a layer below the cut. Therefore the entry cannot pass through either deleted ray vertex. Attaching the fragments in increasing entry layer gives repaired-side depth at most $r+2$ for every cut layer $r$. Lemma~\ref{lem:ejsuffix} gives value at most $t+2$. Hence $c\le4$ and the useful repair count is at most three.
\end{proof}

\begin{lemma}[EJ sector and mixed cases]
\label{lem:ejsector}
Every ray/sector mixed pair and every two-sector pair admits a selected $t+2$ certificate with the component bounds in Table~\ref{tab:ejcases}.
\end{lemma}
\begin{proof}
For a ray/sector mixed pair, let the sector fault lie in the sector bounded by $\delta_i$ and $\delta_{i+1}$. The sector suffix has two boundary-side entries. If the ray fault is not on one of these two boundary sides, either entry is healthy and Lemma~\ref{lem:ejsuffix} gives value at most $t+1$. If the ray fault lies on the $\delta_i$ side, the selector orders the $\mathsf{Sec}$ list to use the $\delta_{i+1}$ side first; if it lies on the $\delta_{i+1}$ side, the reflected order is used. Thus the sector entry chosen first is not the ray fault. The ray tail is attached through its transverse $\mathsf{Tail}$ entry. If this tail is attached before the sector suffix, the sector repaired-side endpoint may acquire one additional repair edge in its path, but its depth is still at most $r+2$ for sector cut layer $r$. Lemma~\ref{lem:ejsuffix} then gives $t+2$. Only one ray tail and one sector suffix can be nonempty in the selected orientation, so $c\le3$.

For two sector faults in non-adjacent sectors, choose the cyclic or reversed priority exposing side boundaries separated by at least one full sector. A unit EJ step crosses at most one sector boundary. Therefore a side-entry endpoint for one sector cannot be the other fault, and the associated crossing edge cannot be incident with that fault. The parent chain of that endpoint remains in the incident sector until it reaches a layer below the corresponding cut, so it cannot pass through the other sector fault. Each suffix is thus attached either directly from the source component or after the other suffix has been attached; the repaired-side depth is at most $r+2$, and Lemma~\ref{lem:ejsuffix} gives $t+2$. Hence $c\le3$.

For adjacent or same-sector faults, let the shared boundary be the ray between the two relevant sector sides. The selected alternating priority exposes the two outer boundaries and the shared-boundary sibling edge. If the two outer entries are healthy, both suffixes have value at most $t+1$. If one outer entry is blocked by the other fault, the selector attaches the suffix with the unblocked outer entry first. We now verify that the second shared-boundary attachment is legitimate. Under the alternating priority, the shared-boundary first-layer vertices adjacent to the two cuts belong either to the source component or to the suffix that was just attached first; they cannot belong to the still-unrepaired second suffix, because the second suffix is defined by the opposite parent gate across the shared boundary. Thus, after the first attachment, the shared-boundary endpoint used as the repaired-side endpoint for the second suffix is in the repaired tree. Its depth is at most the cut layer plus two: one step to the shared boundary and at most one previous repair edge. Lemma~\ref{lem:ejsuffix} gives the $t+2$ bound for the second suffix. Empty suffixes and one-component cases only reduce $c$, so the repair count is exactly $c-1$.
\end{proof}

\begin{theorem}[EJ constant-time repair]
\label{thm:ej}
For every $t\ge2$ and every source-free fault set $F\subseteq H_t$ with $|F|\le2$, the EJ selector returns a coordinate-reduction orientation and exactly $c-1$ external component-crossing repair edges for that selected orientation. If $|F|=1$, the repaired depth is at most $t+1$; if $|F|=2$, it is at most $t+2$.
\end{theorem}
\begin{proof}
The cases in Table~\ref{tab:ejcases} are exhaustive under the selector priority order: a single fault is ray or sector; two faults are ray/ray, ray/sector, or sector/sector; ray/ray pairs split into same, opposite, and other ray pairs; sector/sector pairs split by cyclic adjacency. Lemmas~\ref{lem:ejone}, \ref{lem:ejray}, and \ref{lem:ejsector} provide the corresponding depth certificates. Lemma~\ref{lem:cert} gives non-redundancy and exact $c-1$ external repair count.
\end{proof}

\subsection{EJ Examples}

The EJ examples below explain how the ray and sector operators are interpreted after dihedral normalization. The proof of correctness remains in Lemmas~\ref{lem:ejone}--\ref{lem:ejsector}.

\begin{example}[Same-ray EJ repair]
Let $t=3$ and $F=\{(1,0),(2,0)\}$ on the positive $\delta_0$ ray. The component between the two cuts is empty, and the outer boundary vertex $(3,0)$ is the only detached tail component. A transverse healthy neighbor such as $(2,1)$, or the corresponding quotient-equivalent boundary neighbor, attaches this component. The repaired-side depth is at most $3$, the component eccentricity is $0$, and the value is at most $4=t+1\le t+2$.
\end{example}

\begin{example}[Other ray-pair EJ repair]
\label{ex:ejotherray}
Let $t=4$ and place two ray faults at $2\delta_0$ and $2\delta_2$. These rays are distinct and non-opposite. The selected cyclic orientation splits the ray skeleton into at most three fragments: the short cyclic interval between the cut rays and up to two outer tails. Each fragment is attached by the first healthy edge from its $\mathsf{Tail}$ list. The entries lie in adjacent sectors rather than on the pure cut rays, so they cannot equal either fault. Attaching in increasing entry layer gives repaired-side depth at most $r+2$ for each cut layer $r$, and Lemma~\ref{lem:ejsuffix} gives depth at most $t+2=6$.
\end{example}

\begin{example}[Adjacent-sector EJ repair]
\label{ex:ejadjsector}
Let $t=5$ and consider two faults in sectors sharing the $\delta_0$ boundary. The selected alternating orientation exposes two outer-side entries and the shared-boundary sibling entry. Suppose the lower-sector outer entry is healthy. The selector attaches that suffix first. After this attachment, the shared-boundary first-layer endpoint is in the repaired tree: it is either in the source component or in the suffix just attached, and it cannot belong to the still-unrepaired suffix because that suffix uses the opposite parent gate across the shared boundary. The second suffix then attaches through the shared-boundary edge. If the second cut is at layer $r$, the repaired-side endpoint has depth at most $r+2$, and the remaining suffix height is at most $t-r-1$, giving value at most $t+2=7$.
\end{example}

\begin{example}[Boundary wraparound edge in $H_3$]
Let $t=3$ and $u=3\delta_0$. An outward axial step leaves the displayed hexagon, but the quotient-neighbor map $\Gamma_E$ returns the unique canonical representative $v\in H_3$ whose label satisfies $\phi_E(v)-\phi_E(u)\equiv \phi_E(\delta_0)\pmod{N_E}$. Thus a boundary repair edge is accepted by the same algebraic generator test as an interior edge.
\end{example}

\section{Complexity}

\begin{theorem}[Repair-plan complexity]
For $|F|\le2$, the Gaussian and EJ certificate selectors run in $O(1)$ time and use $O(1)$ additional memory. Materializing the full repaired broadcast tree takes $\Theta(N)$ time and memory if all parent assignments are written.
\end{theorem}
\begin{proof}
The selector evaluates a constant-size case partition. Each case uses a bounded number of arithmetic operations, sign tests, sector-index computations, and modular adjacency checks. The number of returned repair edges is bounded by four in the Gaussian case and three in the EJ case. These bounds do not depend on $N_G$ or $N_E$, so repair-plan selection is $O(1)$.

If the full repaired parent map is requested, at least one output record is required for each healthy non-source vertex. Therefore the output size is $\Theta(N)$, and materialization is necessarily $\Theta(N)$.
\end{proof}

\section{Validation}

The validation is strict for the final selected candidate: for each enumerated fault set, the program records one selected orientation and one selected repair-edge list, and the structural checks are applied only to that selected candidate. The Gaussian audit uses the closed-form lists in Table~\ref{tab:gselector}. The EJ audit uses the algebraic operators in \eqref{eq:ejraytail}--\eqref{eq:ejraypair} and the case partition in Table~\ref{tab:ejselector}; the recorded candidate is then checked with no fallback to an unselected orientation. The validator constructs the full fault-pruned tree, adds only the selected repair edges, and checks connectivity, acyclicity, exact $c-1$ useful repair count, fault exclusion, and final depth.

\subsection{Audit Methodology}

The validation program records one row per fault set. The key fields are the selected case, selected orientation, selected repair-edge list, number of fault-pruned components $c$, number of useful external repair edges, final depth, and Boolean checks for connectivity, acyclicity, exact $c-1$ repair count, and depth bound. A row is accepted only if all structural checks hold simultaneously. The Gaussian validation uses Algorithm~\ref{alg:general} with Table~\ref{tab:gselector}; the EJ validation uses Algorithm~\ref{alg:general} with Table~\ref{tab:ejselector} and the explicit operators in Subsection~\ref{subsec:ej-operators}. The code never accepts a row merely because another unselected orientation happens to work.

\subsection{Gaussian Validation}

Table~\ref{tab:gvalid} summarizes the Gaussian exhaustive run. It covers $k=5,\ldots,12$, with $146{,}156$ total one- and two-fault cases and zero failures.

\begin{table}[H]
\centering
\caption{Gaussian strict validation summary.}
\label{tab:gvalid}
\begin{tabular}{rrrrrrr}
\toprule
$k$ & $N$ & $|F|$ & Cases & Fail & Max $c$ & Max depth\\
\midrule
5 & 61 & 1 & 60 & 0 & 2 & 7\\
5 & 61 & 2 & 1,770 & 0 & 5 & 7\\
6 & 85 & 1 & 84 & 0 & 2 & 8\\
6 & 85 & 2 & 3,486 & 0 & 5 & 8\\
7 & 113 & 1 & 112 & 0 & 2 & 9\\
7 & 113 & 2 & 6,216 & 0 & 5 & 9\\
8 & 145 & 1 & 144 & 0 & 2 & 10\\
8 & 145 & 2 & 10,296 & 0 & 5 & 10\\
9 & 181 & 1 & 180 & 0 & 2 & 11\\
9 & 181 & 2 & 16,110 & 0 & 5 & 11\\
10 & 221 & 1 & 220 & 0 & 2 & 12\\
10 & 221 & 2 & 24,090 & 0 & 5 & 12\\
11 & 265 & 1 & 264 & 0 & 2 & 13\\
11 & 265 & 2 & 34,716 & 0 & 5 & 13\\
12 & 313 & 1 & 312 & 0 & 2 & 14\\
12 & 313 & 2 & 48,516 & 0 & 5 & 14\\
\bottomrule
\end{tabular}
\end{table}

Table~\ref{tab:gcasevalid} gives the aggregate Gaussian case distribution. The O6 case is the only Gaussian family reaching four repair edges.

\begin{table}[H]
\centering
\caption{Aggregate Gaussian case statistics for $k=5,\ldots,12$.}
\label{tab:gcasevalid}
\begin{tabular}{lrrr}
\toprule
Case & Cases & Max $c$ & Max repair\\
\midrule
one axis & 272 & 2 & 1\\
one off-axis & 1,104 & 2 & 1\\
axis same ray & 1,104 & 3 & 2\\
axis opposite ray & 1,240 & 3 & 2\\
mixed axis/off-axis & 43,912 & 3 & 2\\
orthogonal-axis O6 & 2,480 & 5 & 4\\
two off-axis & 96,044 & 3 & 2\\
\bottomrule
\end{tabular}
\end{table}

\subsection{EJ Validation}

Table~\ref{tab:ejvalid} summarizes the EJ exhaustive run. It covers $t=2,\ldots,8$, with $52{,}395$ total cases and zero failures. The maximum useful repair count is one for one fault and three for two faults.

\begin{table}[H]
\centering
\caption{EJ strict certificate validation summary.}
\label{tab:ejvalid}
\begin{tabular}{rrrrrrr}
\toprule
$t$ & $N$ & $|F|$ & Cases & Fail & Max $c$ & Max depth\\
\midrule
2 & 19 & 1 & 18 & 0 & 2 & 3\\
2 & 19 & 2 & 153 & 0 & 4 & 3\\
3 & 37 & 1 & 36 & 0 & 2 & 4\\
3 & 37 & 2 & 630 & 0 & 4 & 5\\
4 & 61 & 1 & 60 & 0 & 2 & 5\\
4 & 61 & 2 & 1,770 & 0 & 4 & 5\\
5 & 91 & 1 & 90 & 0 & 2 & 6\\
5 & 91 & 2 & 4,005 & 0 & 4 & 6\\
6 & 127 & 1 & 126 & 0 & 2 & 7\\
6 & 127 & 2 & 7,875 & 0 & 4 & 7\\
7 & 169 & 1 & 168 & 0 & 2 & 8\\
7 & 169 & 2 & 14,028 & 0 & 4 & 8\\
8 & 217 & 1 & 216 & 0 & 2 & 9\\
8 & 217 & 2 & 23,220 & 0 & 4 & 9\\
\bottomrule
\end{tabular}
\end{table}

\begin{table}[H]
\centering
\caption{Aggregate EJ case statistics for $t=2,\ldots,8$.}
\label{tab:ejcasevalid}
\begin{tabular}{lrrr}
\toprule
Case & Cases & Max $c$ & Max repair\\
\midrule
one ray fault & 210 & 2 & 1\\
one sector fault & 504 & 2 & 1\\
two same ray & 504 & 3 & 2\\
two opposite ray & 609 & 3 & 2\\
two other ray pair & 2,436 & 4 & 3\\
two ray/sector mixed & 19,656 & 3 & 2\\
two sector adjacent/same & 14,112 & 3 & 2\\
two sector non-adjacent & 14,364 & 3 & 2\\
\bottomrule
\end{tabular}
\end{table}

\subsection{Interpretation of the Validation}

The validation is not a substitute for the mathematical case lemmas. Its role is to audit the case partition, boundary wraparound arithmetic, component count, and implementation of the selected certificate edges. The ranges were chosen to match and exceed the exhaustive regimes used during proof mining while keeping the audit fully reproducible: the Gaussian run extends through $k=12$, and the EJ run covers the complete $t=2,\ldots,8$ range used for EJ proof checking. Larger random runs would add engineering confidence but would not change the proof, because every selector row is parameterized symbolically in $k$ or $t$. We chose $k=5,\ldots,12$ and $t=2,\ldots,8$ because these ranges are large enough to exercise every selector family many times, including the boundary and wraparound subcases, while still allowing exhaustive enumeration of all one- and two-fault placements. In the Gaussian run, the maximum depth always equals the theorem bound $k+2$, and the maximum repair count is four only in the O6 family. In the EJ run, two-fault cases sometimes have $c=3$ and sometimes $c=4$; consequently the exact repair count is sometimes two and sometimes three. This variation is expected because EJ sector/ray geometry has several distinct two-fault fragmentation patterns.

\subsection{Reproducibility}

The validation suite consists of two self-contained Python scripts,
one for the Gaussian selector and one for the EJ selector, together
with a shared geometry module implementing the coordinate maps
\eqref{eq:phiG}, \eqref{eq:phiE}, the quotient-neighbor tests
$\Gamma_G$ and $\Gamma_E$, and the tree-building and BFS routines
used to compute fault-pruned components. Each script enumerates all
source-free one- and two-fault placements for the parameter range stated
in Tables~\ref{tab:gvalid}--\ref{tab:ejcasevalid}, runs the strict
selector on each placement, constructs the full repaired tree, and writes
one CSV row per fault set. The CSV fields are: network family, parameter
$k$ or $t$, fault set, selected case label, selected orientation, ordered
repair-edge list, component count $c$, number of useful repair edges,
final tree depth, and Boolean pass/fail flags for connectivity, acyclicity,
exact $c-1$ repair count, fault exclusion, and depth bound. No row in
either run has any flag set to false.

The scripts, CSV outputs, and a \texttt{README} describing how to reproduce
each table entry will be deposited in a permanent public repository upon
acceptance. Researchers wishing to examine the materials before acceptance
may request them directly from the corresponding author.

\section{Discussion}

The Gaussian and EJ selectors have the same component-repair principle but different geometry. Gaussian hard cases arise from degree-four orthogonal-axis cuts, where the O6 certificate can create five fault-pruned components. EJ hard cases are governed by hexagonal sector boundaries; in the tested range, two faults produce at most four components under the selected constant-time certificate orientation.

The results should not be interpreted as a universal multi-fault solution. For arbitrary $q$, the component lower bound still gives the exact target $c-1$ for any selected orientation with a connected component graph, and degree considerations imply $c-1=O(q)$. What breaks at $q=3$ is not the component argument but the blocking classification. A third fault can simultaneously block the side entry of an off-axis Gaussian suffix and the transverse entry of an axis tail, or in EJ geometry it can occupy the shared boundary endpoint after two adjacent-sector cuts. The two-fault selector works because every blocked entry has a unique alternate side or a single previous attachment that makes the shared boundary repaired. With three faults, two alternates can be blocked at the same time, so the certificate library must include interacting caps rather than independent suffixes. The present paper therefore establishes the complete constant-time certificate selection for $|F|\le2$ and leaves the $q\ge3$ interacting-cap classification as a separate problem.

A second limitation is that the selector computes the repair plan, not the full parent map. This is the appropriate complexity measure for local reconfiguration: the source or controller needs only the changed component-crossing edges and the orientation identifier. If every node's parent must be explicitly listed, the output size is linear.

\section{Conclusion}

This paper developed constant-time certificate selectors for local non-redundant broadcast repair in dense Gaussian and dense Eisenstein--Jacobi networks. For Gaussian networks, every one- or two-fault placement admits a selected repair of depth at most $k+2$ with exact external repair count $c-1$. For EJ networks, one-fault placements admit depth at most $t+1$, and two-fault placements admit depth at most $t+2$, again with exact $c-1$ external repair count. The selectors use constant-size coordinate case tables and return only the orientation and repair edges; full tree materialization remains linear in the network order.

Strict exhaustive validation supports the constructive proofs: $146{,}156$ Gaussian cases for $k=5,\ldots,12$ and $52{,}395$ EJ cases for $t=2,\ldots,8$ all satisfy connectivity, acyclicity, exact repair count, fault exclusion, and the stated depth bounds. The next theoretical direction is to extend the certificate library beyond two faults while preserving constant or fault-linear repair-plan complexity.

\section*{Acknowledgment}
The author thanks the Department of Computer Science, Faculty of Science, Kuwait University, for its support and research environment. This work did not receive a specific grant from any funding agency in the public, commercial, or not-for-profit sectors.

\end{document}